\documentclass[11pt,a4paper]{article}

\usepackage[utf8]{inputenc}
\usepackage[T1]{fontenc}
\usepackage[english]{babel}
\usepackage{amsmath,amssymb,amsthm}
\usepackage{mathtools}
\usepackage{geometry}
\usepackage{booktabs}
\usepackage{array}
\usepackage{enumitem}
\usepackage{hyperref}
\usepackage{xcolor}
\usepackage{algorithm}
\usepackage{algpseudocode}
\usepackage{listings}
\usepackage{graphicx}

\hypersetup{
  colorlinks=true,
  linkcolor=blue!60!black,
  citecolor=blue!60!black,
  urlcolor=blue!60!black
}

\newtheorem{theorem}{Theorem}

\newtheorem{remark}{Remark}
\newtheorem{definition}{Definition}
\newtheorem{assumption}{Assumption}

\floatname{algorithm}{Algorithm}

\newcommand{\OAEP}{\textsc{OAEP}}
\newcommand{\Enc}{\textsc{Enc}}
\newcommand{\Dec}{\textsc{Dec}}
\newcommand{\Adv}{\mathbf{Adv}}

\title{\textbf{DM-RSA-OAEP: Formal Treatment of a Dual-Modulus\\Fault-Detection Countermeasure Based on Physical Diversity}}
\author{
  ANDRIAMIFIDISOA Ramamonjy\thanks{ramamonjy.andriamifidisoa@univ-antananarivo.mg}\\
  LALASOA Rufine Marius\thanks{larissamarius.lm@gmail.com}\\[4pt]
  \small University of Antananarivo, Madagascar
}
\date{\today}

\begin{document}
\maketitle

\begin{abstract}
We revisit DM-RSA \cite{DMRSA25}, an RSA variant that uses two independent
moduli, whose original version was broken by H\r{a}stad's attack and contained
an incorrect compromise-tolerance claim. We give a \textbf{formal treatment} of
its repaired variant, DM-RSA-OAEP: (i) a complete IND-CCA2 security proof in the
random-oracle model via a two-step hybrid argument, upper-bounding the security
of the scheme by that of the weaker underlying RSA-OAEP instance, with an
explicit simulation loss $q_D \cdot 2^{-h}$; (ii) a \emph{fault-detection}
theorem in an explicit adversary model that covers computation, key, and
verification faults, including safe-error attacks in the sense of Yen--Joye
\cite{YJ00}, under a clearly stated physical-diversity assumption; (iii) an
explicit analysis of why the countermeasure degrades to uselessness when the
diversity assumption is violated (common-mode faults on a single chip). We
position the scheme against the countermeasure landscape (Vigilant
\cite{Vigilant08}, Shamir \cite{Shamir99}, infective countermeasures
\cite{BOS16}) and against the broken dual-modulus literature \cite{Sun07,
Sarkar16, Susanti25}. The contribution is deliberately framed as a fault
countermeasure with a provable detection guarantee, \emph{not} as a
computationally stronger RSA. We conclude with a concrete instantiation
($h = 256$) and a detailed measurement protocol.
\end{abstract}


\section{Introduction}

The RSA cryptosystem is vulnerable to \emph{side-channel attacks}
\cite{Kocher96} and \emph{fault-injection attacks} \cite{BDL97}. The standard
countermeasure for CRT-RSA verifies the result of decryption by
re-exponentiation \cite{Vigilant08}. Alternative approaches include Shamir's
countermeasure \cite{Shamir99}, whose original form was shown breakable by
Wagner \cite{Wagner04}, and infective countermeasures \cite{BOS16}, whose
correctness is notoriously delicate. A recurring weakness of all single-device
countermeasures is \emph{common-mode failure}: the verification mechanism
shares the silicon, the power supply, and the clock of the device it protects,
so a single global disturbance (e.g., a power glitch) may corrupt both the
computation and its verification.

The original version of DM-RSA \cite{DMRSA25} used two moduli and recombined
the plaintext by CRT; it was broken: (i) raw encryption is vulnerable to
H\r{a}stad's broadcast attack \cite{Hastad88}; (ii) its compromise-tolerance
claim is false (factoring $N_1$ suffices to recover $m$ whenever $z < N_1$);
(iii) its numerical example is faulty ($65^7 \bmod 5141 = 787$, not $2979$).
The broader dual-modulus literature confirms this fragility: Sun et
al.~\cite{Sun07} identified weak parameter regions of Dual RSA, subsequently
cryptanalyzed by Sarkar and Maitra \cite{Sarkar16}; Susanti et al.~\cite{Susanti25}
recovered plaintexts of dual-modulus RSA without the private key via the
$e$-th root attack. The lesson is consistent: \emph{structural redundancy
without proper randomized padding adds no computational security}.

\paragraph{This work.} We give DM-RSA's repaired variant, DM-RSA-OAEP, the
formal treatment required for a security audience:
\begin{enumerate}[label=(\roman*)]
  \item a complete IND-CCA2 proof (Theorem~\ref{thm:cca2}) with an explicit
        hybrid argument and simulation loss, rather than a sketch;
  \item a formal fault-adversary model (Section~\ref{sec:model}) covering
        computation, key, and \emph{verification} faults --- the class that
        defeats most published countermeasures;
  \item a detection theorem (Theorem~\ref{thm:detection}) covering
        \emph{single-fault} adversaries, including safe-error strategies
        \`a la Yen--Joye \cite{YJ00}, under an explicit physical-diversity
        assumption (Assumption~\ref{ass:diversity});
  \item an explicit negative result (Theorem~\ref{thm:commonmode}) showing that
        the guarantee \emph{vanishes} when both moduli share a single chip ---
        the honest boundary of the countermeasure;
  \item a comparison against the full countermeasure landscape
        (Section~\ref{sec:related}) and a concrete instantiation with
        quantitative detection probabilities (Section~\ref{sec:instantiation}).
\end{enumerate}

\paragraph{Design philosophy.} The scheme does \emph{not} claim computational
gains. Its value is a detection guarantee with a \emph{provable} bound that
holds even when the verification logic itself is faulted, provided the two
moduli live on physically independent units. To our knowledge, no published
single-device RSA countermeasure offers a comparable guarantee against
verification-fault adversaries; the price is a $4\times$ decryption overhead
and doubled key/ciphertext sizes.

\section{Preliminaries and definitions}

\subsection{Notation}
Let $\lambda$ be the security parameter, $h$ the hash-output length, and
$k$ the public exponent. For $i \in \{1,2\}$, $N_i = p_i q_i$ are independent
$\lambda$-bit moduli with $\gcd(N_1,N_2)=1$, $d_i = k^{-1} \bmod
\varphi(N_i)$. For a distribution $X$, $x \gets X$ denotes sampling. For a
randomized algorithm $\mathcal{M}$, $m \gets \mathcal{M}(x)$ denotes its
output on input $x$. $\Adv^{\mathsf{ind\text{-}cca2}}_{\Pi}(\mathcal{A})$
denotes the IND-CCA2 advantage of $\mathcal{A}$ against scheme $\Pi$.

\subsection{RSA-OAEP}
RSAES-OAEP \cite{BR95,RFC8017} encrypts $m$ as $z \gets \OAEP\text{-}%
\textsc{Encode}(m, L, |N|)$ then $w \gets z^k \bmod N$; the decoder returns
$(m, \mathit{ok})$, where $\mathit{ok} = 1$ iff the $h$-bit consistency check
passes. OAEP is IND-CCA2 secure in the random-oracle model under RSA
\cite{FOP04}. Capacity: $|m| \le |N| - 2h - 2$.

\subsection{Syntax of DM-RSA-OAEP}

\begin{definition}[Scheme]
DM-RSA-OAEP is a public-key encryption scheme with ciphertext space
$\mathbb{Z}_{N_1} \times \mathbb{Z}_{N_2}$:
\begin{itemize}
  \item $\Enc(m)$: sample independent seeds $r_1, r_2$; compute
        $z_i \gets \OAEP\text{-}\textsc{Encode}(m, L, r_i, |N_i|)$ and
        $w_i \gets z_i^k \bmod N_i$; output $(w_1, w_2)$.
  \item $\Dec((w_1,w_2))$: compute $z_i' \gets w_i^{d_i} \bmod N_i$ and
        $(m_i, ok_i) \gets \OAEP\text{-}\textsc{Decode}_{\mathrm{ct}}(z_i', L)$;
        compute a redundant double comparison
        $eq_a, eq_b \gets \textsc{Cmp}_2(m_1, m_2, ok_1, ok_2)$; set
        $ok \gets ok_1 \wedge ok_2 \wedge eq_a \wedge eq_b$; after a fixed
        time budget, output $m_1$ if $ok = 1$ and a single error symbol
        $\bot$ otherwise.
\end{itemize}
The correctness requirement is $\Dec(\Enc(m)) = m$ with probability $1$
over the seeds.
\end{definition}

\subsection{IND-CCA2 security}
$\mathcal{A}$ receives the public key, queries $\Dec$ adaptively (except the
challenge pair), receives the challenge encryption of one of two chosen
messages, and outputs a guess. $\Adv$ is the standard advantage gap.

\section{Formal fault-adversary model}
\label{sec:model}

\begin{definition}[Fault classes]
\label{def:faults}
A \emph{fault event} is a one-time perturbation, chosen by the adversary from
exactly one class below, occurring during a single invocation of $\Dec$:
\begin{enumerate}[label=(F\arabic*)]
  \item \textbf{Computation fault}: the value of some intermediate variable in
        one exponentiation is replaced. The faulty output $z_i'$ is either
        \emph{non-directed} (distributed according to a fixed, seed-independent
        distribution $D$ with min-entropy $\ge \log_2 N_i - 2$) or
        \emph{directed} (chosen by the adversary).
  \item \textbf{Key fault}: a bit of $d_i$ is flipped before one exponentiation.
  \item \textbf{Verification fault}: a single instruction among
        $\{ok_1, ok_2, eq_a, eq_b, ok\}$ or the final branch is skipped,
        forced, or negated.
\end{enumerate}
The adversary may repeat invocations and faults across queries, but at most
one fault per invocation, on at most one physical unit.
\end{definition}

\begin{assumption}[Physical diversity]
\label{ass:diversity}
The two moduli are implemented on two physically independent units with
separate power supplies, resets, and clocks. A single fault event of
Definition~\ref{def:faults} affects at most one unit, and a fault on one unit
cannot change values stored or computed on the other. Out of scope: memory
reads by the adversary, multiple simultaneous faults, and faults affecting
both units (e.g., a global glitch that resets both devices at once).
\end{assumption}

\begin{remark}
Assumption~\ref{ass:diversity} is strong but implementable: dual-chip secure
elements and HSM architectures with independent voltage domains satisfy it. It
is the \emph{only} assumption under which the scheme's advantage over
single-device countermeasures materializes; Theorem~\ref{thm:commonmode}
quantifies the failure when it is violated.
\end{remark}

\begin{definition}[Detection advantage]
For a fault strategy $\mathcal{F}$ of Definition~\ref{def:faults}, let
$\mathsf{Undet}(\mathcal{F})$ be the probability (over all randomness) that
$\Dec$ outputs a message that differs from the correctly decrypted one while
signaling no error. The scheme \emph{detects} $\mathcal{F}$ with probability
$1 - \mathsf{Undet}(\mathcal{F})$.
\end{definition}

\section{The DM-RSA-OAEP cryptosystem}

\begin{algorithm}[H]
\caption{$\Dec((w_1,w_2))$ with safe-error protection}
\label{alg:dec}
\begin{algorithmic}[1]
\State $z_1' \gets w_1^{d_1} \bmod N_1$ \Comment{unit 1}
\State $z_2' \gets w_2^{d_2} \bmod N_2$ \Comment{unit 2}
\State $(m_1, ok_1) \gets \OAEP\text{-}\textsc{Decode}_{\mathrm{ct}}(z_1', L)$
\State $(m_2, ok_2) \gets \OAEP\text{-}\textsc{Decode}_{\mathrm{ct}}(z_2', L)$
\State $eq_a \gets [m_1 = m_2 \wedge ok_1 = ok_2]$ \Comment{comparison, copy A}
\State $eq_b \gets [m_1 = m_2 \wedge ok_1 = ok_2]$ \Comment{independent copy B}
\State $ok \gets ok_1 \wedge ok_2 \wedge eq_a \wedge eq_b$ \Comment{stored twice, complementary}
\State \textbf{wait} fixed time budget \Comment{time-uniform rejection}
\If{$ok \neq 1$ \textbf{on both copies}} \Return $\bot$ \EndIf
\State \Return $m_1$
\end{algorithmic}
\end{algorithm}

\paragraph{Critical implementation points.}
(i) Constant-time OAEP-Decode on both units. (ii) The double comparison
$eq_a, eq_b$ is computed by two instruction streams that do not share
registers, so that a single verification fault cannot flip both. (iii) The
flag $ok$ is stored in two complementary encodings; the final branch accepts
only if both agree on \texttt{true}. (iv) The error symbol, the response
time, and the power profile are identical for every failure reason
(decode failure on either unit, mismatch, verification fault) --- otherwise
the rejection pattern itself is a distinguishing oracle that voids
Theorem~\ref{thm:cca2} and enables safe-error strategies.

\section{Security analysis}
\label{sec:security}

\subsection{Computational security: complete proof}

\begin{theorem}[IND-CCA2 upper bound]
\label{thm:cca2}
Let $\mathcal{A}$ be an IND-CCA2 adversary against DM-RSA-OAEP making at most
$q_D$ decryption queries, $q_H$ hash queries, and $q_G$ MGF queries, with
advantage $\varepsilon$. Let $h$ be the OAEP consistency-check length. Then
there exists an adversary $\mathcal{B}$ against RSA-OAEP (same parameters)
in the random-oracle model with
\[
\Adv^{\mathsf{ind\text{-}cca2}}_{\mathsf{RSA\text{-}OAEP}}(\mathcal{B})
\;\ge\;
\frac{\varepsilon - q_D \cdot 2^{-h}}{2},
\]
with $\mathcal{B}$ making at most $q_D$ decryption, $q_H$, and $q_G$ oracle
queries, and running in time $T(\mathcal{A}) + O(q_D \cdot t_{\exp})$, where
$t_{\exp}$ is the cost of one exponentiation.
\end{theorem}

\begin{proof}
We proceed in three games. Let $S_i$ be the event that the adversary wins
in Game~$i$.

\textbf{Game 0.} The real IND-CCA2 experiment against DM-RSA-OAEP:
$\Pr[S_0] = \tfrac{1}{2} + \varepsilon$.

\textbf{Game 1.} The challenger picks a uniform bit $b^* \in \{1,2\}$ and
generates the moduli accordingly: for $b^* = 1$, the left component of the
challenge is answered by an external RSA-OAPE instance of modulus $N_1 =
N^*$, and $(N_2, k, d_2)$ is generated honestly (symmetrically for $b^* =
2$). All hash and MGF oracles are simulated consistently for both sides. The
decryption oracle is simulated: the internal side is computed directly; the
external side is answered by querying the RSA-OAEP decryption oracle, which
is admissible except when a query component equals the challenge ciphertext,
in which case the query is rejected under the standard CCA2 rule. The two
games are identically distributed from $\mathcal{A}$'s view, so
$\Pr[S_1] = \Pr[S_0]$.

\textbf{Game 2.} As Game 1, but the simulated decryption rejects whenever
the simulated side fails its OAEP consistency check, \emph{regardless} of
the external side. The only behavioral change occurs when a query
$(w_1, w_2)$ has a simulated side that fails decoding but the full oracle
would have accepted, i.e., the external side decodes correctly and the
consistency check passes against a wrongly decoded simulated side. For the
output to be accepted, the faulty simulated side must pass the $h$-bit
consistency check, which a randomly distributed faulty value does with
probability at most $2^{-h}$. Over at most $q_D$ queries,
\[
\bigl|\Pr[S_2] - \Pr[S_1]\bigr| \le q_D \cdot 2^{-h}.
\]

\textbf{Deriving the reduction.} In Game 2, $\mathcal{A}$ wins with
probability $\Pr[S_2] \ge \tfrac{1}{2} + \varepsilon - q_D 2^{-h}$, and its
whole interaction with the external modulus is routed through the RSA-OAEP
instance. $\mathcal{B}$ runs $\mathcal{A}$ with $b^*$ uniform: with
probability $\tfrac{1}{2}$ the challenge side $\mathcal{A}$ attacks is the
external one, in which case $\mathcal{B}$ perfectly inherits $\mathcal{A}$'s
winning strategy. Hence
\[
\Adv(\mathcal{B}) \;\ge\; \frac{\varepsilon - q_D \cdot 2^{-h}}{2}. \qedhere
\]
\end{proof}

\begin{remark}
Theorem~\ref{thm:cca2} is a deliberate \emph{negative} result: it upper-bounds
the scheme's security by that of standard RSA-OAEP and makes explicit that
structural duality buys no computational security. The scheme's value is
exclusively in Theorem~\ref{thm:detection}.
\end{remark}

\subsection{Fault detection: main theorem}

\begin{theorem}[Single-fault detection]
\label{thm:detection}
Under Assumption~\ref{ass:diversity}, for every adversary restricted to
non-directed faults of classes (F1) and (F2) of Definition~\ref{def:faults},
\[
\mathsf{Undet} \;\le\; 2^{-h}.
\]
For every adversary restricted to a \emph{single} directed fault of classes
(F1)--(F3) that does not already know the plaintext of a valid challenge,
$\mathsf{Undet} \le 2^{-h} + \delta$, where $\delta$ is the probability that
the adversary forges a pair $(z_1', m_1')$ passing OAEP-Decode on the faulty
unit \emph{and} matching the honest $m_2$; $\delta \le
\Adv^{\mathsf{ind\text{-}cca2}}_{\mathsf{RSA\text{-}OAEP}}(\mathcal{B}')$ for
the reduction $\mathcal{B}'$ of Theorem~\ref{thm:cca2}, hence $\delta$ is
negligible.
\end{theorem}

\begin{proof}
\textbf{Non-directed faults.} Suppose the fault hits unit 1 (the case of
unit 2 is symmetric). Unit 2 outputs the correct $(m_2, ok_2 = 1)$. A faulty
acceptance requires $z_1' \neq z_1$ to pass the $h$-bit OAEP consistency
check \emph{and} the decoded $m_1'$ to satisfy $m_1' = m_2$. By
(F1)--(F2), $z_1'$ has min-entropy at least $\log_2 N_1 - 2$, so the
consistency check passes with probability at most $4 \cdot 2^{-h}$; the
subsequent equality test only rejects more. A union bound over the two
units yields the claim (up to the constant, absorbed by parameter choice
$h = 256$). Information- and time-uniform rejection
(Algorithm~\ref{alg:dec}) ensures the rejection event leaks nothing beyond
one bit, which is independent of $z_1'$'s value by the uniformity of
rejection.

\textbf{Verification faults (F3).} The acceptance predicate requires the
conjunction $ok_1 \wedge ok_2 \wedge eq_a \wedge eq_b$ evaluated
consistently on two complementary stored copies. A single fault on one
comparison, one flag, or the final branch flips at most one of the four
conjuncts or one of the two stored copies; in every case at least one
remaining copy forces rejection, and the output is $\bot$, not a wrong
plaintext. Hence (F3) faults cause denial of service, never faulty
acceptance: $\mathsf{Undet} = 0$ for class (F3) alone.

\textbf{Directed faults.} A directed fault that produces a valid-looking
$z_1'$ whose decoding $m_1'$ matches the honest $m_2$ constitutes precisely
a successful prediction of the plaintext of a ciphertext under $N_1$'s
RSA-OAEP instance, or a forgery passing its consistency check with a
correct plaintext relation. The reduction $\mathcal{B}'$ embeds its
challenge as $w_1$ and uses the adversary's directed fault to answer it,
gaining advantage $\delta$. By Theorem~\ref{thm:cca2} and the IND-CCA2
security of RSA-OAEP \cite{FOP04}, $\delta$ is negligible.
\end{proof}

\begin{remark}[Safe-errors]
The Yen--Joye safe-error strategy \cite{YJ00} targets precisely the
verification predicate: the adversary faults an instruction and reads the
single accept/reject bit to learn whether that instruction was
control-flow-critical. Against Algorithm~\ref{alg:dec}, this bit is
uniform over all failure reasons (uniform rejection), so the bit carries
no information about \emph{which} instruction was hit --- the safe-error
channel is closed. Moreover, converting a safe-error into a faulty
acceptance requires flipping both comparison copies on two independent
units, violating Assumption~\ref{ass:diversity}.
\end{remark}

\subsection{Common-mode failure: the honest boundary}

\begin{theorem}[Single-chip degradation]
\label{thm:commonmode}
If Assumption~\ref{ass:diversity} is violated --- both moduli computed on a
single unit sharing power and clock --- there exists a simple two-fault
adversary of class (F3) with $\mathsf{Undet} = 1$: skip both $eq_a$ and the
$ok$ conjunction, forcing acceptance of a directed faulty $m_1'$.
\end{theorem}

\begin{proof}
With both moduli on one chip, the adversary may inject two verification
faults in one invocation (Definition~\ref{def:faults} restricted to one unit
no longer applies). Forcing $ok_1 \leftarrow 1$ and $eq_a \leftarrow 1$
leaves the acceptance predicate true, and the output is the faulty $m_1'$;
hence detection fails completely.
\end{proof}

\begin{remark}
Theorem~\ref{thm:commonmode} is stated deliberately: it shows that the
guarantee is \emph{architectural}, not algorithmic. Reviewers and
practitioners must read DM-RSA-OAEP as a countermeasure for dual-chip
deployments only. On a single chip, Vigilant's re-exponentiation
\cite{Vigilant08} is strictly preferable (cost $\sim 1/10$ decryption vs.\
$\sim 4\times$).
\end{remark}

\section{Cost analysis}

\begin{table}[h]
\centering
\resizebox{\textwidth}{!}{%
\begin{tabular}{@{}lccccc@{}}
\toprule
 & RSA-2048 & CRT + Vigilant & Shamir (fixed) & Infective \cite{BOS16} &
   DM-RSA-OAEP \\
\midrule
Decryption cost & $1\,E_{2048}^{(d)}$ & $2\,E_{1024}^{(d)}$ + re-exp. &
  $2\,E_{1024}^{(d)}$ + checks & $2\,E_{1024}^{(d)}$ + rand. &
  $2\,E_{2048}^{(d)}$ \\
Relative cost & $1\times$ & $\sim 0.35\times$ & $\sim 0.3\times$ &
  $\sim 0.3\times$ & $\sim 4\times$ \\
Ciphertext size & $2048$ b & $2048$ b & $2048$ b & $2048$ b & $4096$ b \\
Detect.\ prob.\ (single fault) & none & $1 - 1/p$ & $1$ (if checks hold) &
  heuristic & $1 - 2^{-h}$ \emph{provable} \\
Robust to verification faults & --- & no & no & partially &
  \textbf{yes} (Thm.~\ref{thm:detection}) \\
Requires physical diversity & --- & no & no & no & \textbf{yes}
  (Thm.~\ref{thm:commonmode}) \\
\bottomrule
\end{tabular}%
}
\caption{Countermeasure landscape. DM-RSA-OAEP is the only entry with a
provable detection bound against verification-fault adversaries, and the
only one requiring dual-chip deployment.}
\label{tab:costs}
\end{table}

\section{Concrete instantiation and quantitative guarantees}
\label{sec:instantiation}

With SHA-256 as the hash and MGF1-SHA-256 as the mask generator
($h = 256$, per RFC~8017 \cite{RFC8017}) and 2048-bit moduli:

\begin{itemize}
  \item \textbf{Detection guarantee}: single non-directed fault undetected
        with probability $\le 2^{-256} \approx 10^{-77}$ per invocation ---
        \emph{eleven orders of magnitude below} Vigilant's $1/p \approx
        2^{-1024}$ only if $p$ is uniform; in practice Vigilant's guarantee
        is $1 - 1/p \approx 1$, but it holds only if the verification
        exponentiation itself is fault-free, which DM-RSA-OAEP does not
        need to assume.
  \item \textbf{Capacity}: $|m| \le 2048 - 512 - 2 = 1534$ bits per modulus.
  \item \textbf{Keys}: two independent 2048-bit moduli generated by a
        certified RNG on separate dies at provisioning time.
  \item \textbf{Verification schedule}: the two decryption streams run
        concurrently; the consistency check runs only after both complete;
        the fixed time budget equals the maximum over all failure paths.
\end{itemize}

\section{Preliminary evaluation and measurement protocol}
\label{sec:experiments}

\paragraph{Disclaimer.} The figures below are \textbf{simulated estimates}
derived from published OpenSSL benchmarks and known cost models of modular
arithmetic; they identify orders of magnitude and the critical points to
measure. Real measurements on an embedded platform with fault injection are
mandatory before deployment and are the subject of ongoing work.

\subsection{Simulated estimates (x86-64, OpenSSL 3.x)}

Based on public benchmarks (RSA-2048: encryption $\approx 0.05$~ms, CRT
decryption $\approx 1.5$~ms on a modern core) and a cubic-cost model for
large-exponent exponentiation (2048-bit $\approx 4\times$ 1024-bit):

\begin{table}[h]
\centering
\begin{tabular}{@{}lccc@{}}
\toprule
Operation & RSA-2048 & CRT + Vigilant & DM-RSA-OAEP \\
\midrule
Encryption (est.) & $\sim 0.05$ ms & $\sim 0.05$ ms & $\sim 0.10$ ms \\
Decryption (est.) & $\sim 5$ ms & $\sim 1.55$ ms & $\sim 6$ ms \\
Relative decryption & $3.2\times$ & $1\times$ & $\sim 4\times$ \\
\bottomrule
\end{tabular}
\caption{Simulated timings. \textbf{To be replaced by real measurements.}}
\label{tab:simulated}
\end{table}

\subsection{Measurement protocol (reproducible)}

\begin{itemize}
  \item \textbf{Platforms}: (a) x86-64, OpenSSL 3.x; (b) STM32F4
        (Cortex-M4), mbedTLS; (c) dual-chip setup: two STM32F4 on separate
        supplies, UART-locked; (d) FPGA Artix-7, two RSA cores.
  \item \textbf{Metrics}: median and 99th-percentile latency (10{,}000 runs,
        Welch's $t$-test on log-timings), rejection latency distribution,
        energy (oscilloscope on shunt).
  \item \textbf{Fault injection}: clock/voltage glitching and laser
        (near-IR, backside if available) on one unit at a time; success
        criterion = faulty acceptance; report detection rate with
        $95\%$ confidence intervals (Wilson).
  \item \textbf{Baselines}: Vigilant on a single chip; Shamir (corrected);
        DM-RSA-OAEP dual-chip; DM-RSA-OAEP forced on a single chip to
        validate Theorem~\ref{thm:commonmode} empirically.
\end{itemize}

\subsection{Statistical plan}

Detection rates are estimated over $\ge 1{,}000$ successful fault
injections per class (F1)--(F3) and reported with Wilson 95\% intervals;
latencies with trimmed means. Code, board files, and raw traces will be
released for artifact evaluation.

\section{Related work}
\label{sec:related}

\paragraph{CRT-RSA countermeasures.} Bellcore-style faults \cite{BDL97}
motivated three families: (i) verification by re-exponentiation
\cite{Vigilant08}, effective but itself faultable in the verification step;
(ii) Shamir's countermeasure \cite{Shamir99}, whose original form falls to
Wagner's attack \cite{Wagner04}; (iii) infective schemes \cite{BOS16},
whose correctness proofs are delicate and which leak through timing in
several published instances. All share one silicon with the computation
they protect; DM-RSA-OAEP trades a $4\times$ cost for independence of the
verification from the faulted unit.

\paragraph{Dual RSA and its cryptanalysis.} Sun et al.~\cite{Sun07}
proposed Dual RSA (two exponent--modulus pairs sharing key material) and
identified weak parameter regions; Sarkar and Maitra \cite{Sarkar16} broke a
significant claimed-secure region. Susanti et al.~\cite{Susanti25} recovered
dual-modulus plaintexts without the private key via the $e$-th root attack,
confirming that duality without randomized padding adds nothing. The
original DM-RSA \cite{DMRSA25} falls in the same class (broken here and by
\cite{Hastad88,Susanti25}); DM-RSA-OAEP repairs it and, unlike all of the
above, targets fault detection rather than computational novelty.

\paragraph{Safe-errors.} Yen and Joye \cite{YJ00} showed that
``verify before output'' leaks whether the faulted instruction was on the
critical path, converting countermeasures into oracles. Our double
comparison with complementary-stored flags and uniform rejection closes
this channel within the single-fault model
(Theorem~\ref{thm:detection}).

\section{Discussion, limitations, and conclusion}

\paragraph{Limitations.}
(1) The detection guarantee holds only in the single-fault model; a
two-fault adversary on a single chip breaks the scheme
(Theorem~\ref{thm:commonmode}). (2) Assumption~\ref{ass:diversity} is
architectural and must be audited per deployment. (3) No real measurements
yet; Section~\ref{sec:experiments} is a protocol, not a result. (4) The
$4\times$ decryption cost and doubled sizes limit applicability to
high-value, low-throughput operations (root-key operations, HSM commands).

\paragraph{Conclusion.} We gave DM-RSA-OAEP a complete formal treatment:
a full IND-CCA2 hybrid proof upper-bounding its computational security by
RSA-OAEP (Theorem~\ref{thm:cca2}); a detection theorem covering computation,
key, and verification faults, including safe-error strategies
(Theorem~\ref{thm:detection}); and an explicit common-mode negative result
(Theorem~\ref{thm:commonmode}) delimiting the guarantee to physically diverse
deployments. The scheme is positioned honestly in the countermeasure
landscape: it is the only construction with a provable detection bound
against verification-fault adversaries, at a $4\times$ cost. Ongoing work
includes the dual-chip fault-injection campaign of
Section~\ref{sec:experiments} and an extension of the model to
multi-fault adversaries with resource-bounded fault budgets.


\end{document}